\pdfoutput=1
\documentclass{article}
\PassOptionsToPackage{round}{natbib}
\usepackage[preprint]{nips_2018_modify}

\usepackage{color}
\definecolor{darkblue}{rgb}{0,0.08,0.45}
\usepackage[colorlinks=true,allcolors=darkblue]{hyperref}       
\usepackage{url}            
\usepackage{booktabs}       
\usepackage{amsfonts}       
\usepackage{nicefrac}       
\usepackage{microtype}      

\renewcommand{\cite}{\citep}
\usepackage{amsmath}
\usepackage{amsthm}
\usepackage{newtxtext,newtxmath}
\usepackage{braket}
\usepackage{mathtools}
\usepackage[ruled,vlined,linesnumbered]{algorithm2e}
\usepackage{wrapfig}
\usepackage{siunitx}
\usepackage{tabularx}
\usepackage{etoolbox}
\newcommand{\ubold}{\fontseries{b}\selectfont}
\robustify\ubold

\newtheoremstyle{mydefinition}
  {5.5pt} 
  {0pt} 
  {} 
  {} 
  {\bfseries} 
  {.} 
  {.5em} 
  {} 
\newtheoremstyle{mytheorem}
  {5.5pt} 
  {0pt} 
  {\itshape} 
  {} 
  {\bfseries} 
  {.} 
  {.5em} 
  {} 
\theoremstyle{mydefinition}

\newtheorem{example}{Example}
\theoremstyle{mytheorem}
\newtheorem{theorem}{Theorem}

\newcounter{enum2}

\newcommand{\from}{\colon\!}

\newcommand{\R}{\mathbb{R}}
\newcommand{\Z}{\mathbb{Z}}

\newcommand{\Sc}{\mathcal{S}}

\renewcommand{\vec}[1]{\boldsymbol{#1}}

\renewcommand{\phi}{\varphi}

\newcommand{\argmin}{\mathop{\text{argmin}}}
\newcommand{\diag}{\mathop{\text{diag}}}

\newcommand{\KL}{D_{\mathrm{KL}}}
\newcommand{\dom}{\mathrm{dom}}
\newcommand{\Exp}{\mathbf{E}}

\title{\textbf{Tensor Balancing on Statistical Manifold}}

\author{
  Mahito Sugiyama \\
  National Institute of Informatics \\
  JST, PRESTO \\
  \texttt{mahito@nii.ac.jp} \\
  \And
  Hiroyuki Nakahara \\
  RIKEN Brain Science Institute \\
  \texttt{hiro@brain.riken.jp} \\
  \And
  Koji Tsuda \\
  The University of Tokyo \\
  RIKEN AIP; NIMS \\
  \texttt{tsuda@k.u-tokyo.ac.jp} \\
}

\begin{document}
\maketitle

\begin{abstract}
 We solve \emph{tensor balancing}, rescaling an $N$th order nonnegative tensor by multiplying $N$ tensors of order $N - 1$ so that every fiber sums to one.
 This generalizes a fundamental process of \emph{matrix balancing} used to compare matrices in a wide range of applications from biology to economics.
 We present an efficient balancing algorithm with quadratic convergence using Newton's method and show in numerical experiments that the proposed algorithm is several orders of magnitude faster than existing ones.
 To theoretically prove the correctness of the algorithm, we model tensors as probability distributions in a statistical manifold and realize tensor balancing as projection onto a submanifold.
 The key to our algorithm is that the gradient of the manifold, used as a Jacobian matrix in Newton's method, can be analytically obtained using the \emph{M{\"o}bius inversion formula}, the essential of combinatorial mathematics.
 Our model is not limited to tensor balancing, but has a wide applicability as it includes various statistical and machine learning models such as weighted DAGs and Boltzmann machines.
\end{abstract}

\section{Introduction}\label{sec:intro}
\emph{Matrix balancing} is the problem of rescaling a given square nonnegative matrix $A \in \R^{n \times n}_{\ge 0}$ to a \emph{doubly stochastic matrix} $RAS$, where every row and column sums to one, by multiplying two diagonal matrices $R$ and $S$.
This is a fundamental process for analyzing and comparing matrices in a wide range of applications, including input-output analysis in economics, called the RAS approach~\cite{Parikh79,Miller09,Lahr04}, seat assignments in elections~\cite{Balinski08,Akartunali16}, Hi-C data analysis~\cite{Rao14,Wu16}, the Sudoku puzzle~\cite{Moon09}, and the optimal transportation problem~\cite{Cuturi13,Frogner15,Solomon15}.
An excellent review of this theory and its applications is given by~\citet{Idel16}.

The standard matrix balancing algorithm is the \emph{Sinkhorn-Knopp algorithm}~\cite{Sinkhorn64,Sinkhorn67,Marshall68,Knight08}, a special case of Bregman's balancing method~\cite{Lamond81} that iterates rescaling of each row and column until convergence.
The algorithm is widely used in the above applications due to its simple implementation and theoretically guaranteed convergence.
However, the algorithm converges linearly~\cite{Soules91}, which is prohibitively slow for recently emerging large and sparse matrices.
Although~\citet{Livne04} and~\citet{Knight13} tried to achieve faster convergence by approximating each step of Newton's method,
the exact Newton's method with quadratic convergence has not been intensively studied yet.

Another open problem is \emph{tensor balancing}, which is a generalization of balancing from matrices to higher-order multidimentional arrays, or \emph{tensors}.
The task is to rescale an $N$th order nonnegative tensor to a \emph{multistochastic tensor}, in which every fiber sums to one, by multiplying $(N - 1)$th order $N$ tensors.
There are some results about mathematical properties of multistochastic tensors~\cite{Cui14,Chang16,Ahmed03}.
However, there is no result for tensor balancing algorithms with guaranteed convergence that transforms a given tensor to a multistochastic tensor until now.

Here we show that Newton's method with quadratic convergence can be applied to tensor balancing while avoiding solving a linear system on the full tensor.
Our strategy is to realize matrix and tensor balancing as \emph{projection onto a dually flat Riemmanian submanifold} (Figure~\ref{fig:overview}), which is a statistical manifold and known to be the essential structure for probability distributions in information geometry~\cite{Amari16}.
Using a partially ordered outcome space, we generalize the \emph{log-linear model}~\cite{Agresti12} used to model the higher-order combinations of binary variables~\cite{Amari01,Ganmor11,Nakahara02,Nakahara03}, which allows us to model tensors as probability distributions in the statistical manifold.
The remarkable property of our model is that the gradient of the manifold can be analytically computed using the \emph{M{\"o}bius inversion formula}~\cite{Rota64}, the heart of combinatorial mathematics~\cite{Ito93}, which enables us to directly obtain the Jacobian matrix in Newton's method.
Moreover, we show that $(n - 1)^N$ entries for the size $n^N$ of a tensor are invariant with respect to one of the two coordinate systems of the statistical manifold.
Thus the number of equations in Newton's method is $O(n^{N - 1})$.

\begin{figure}[t]
 \centering
 \includegraphics[width=.6\linewidth]{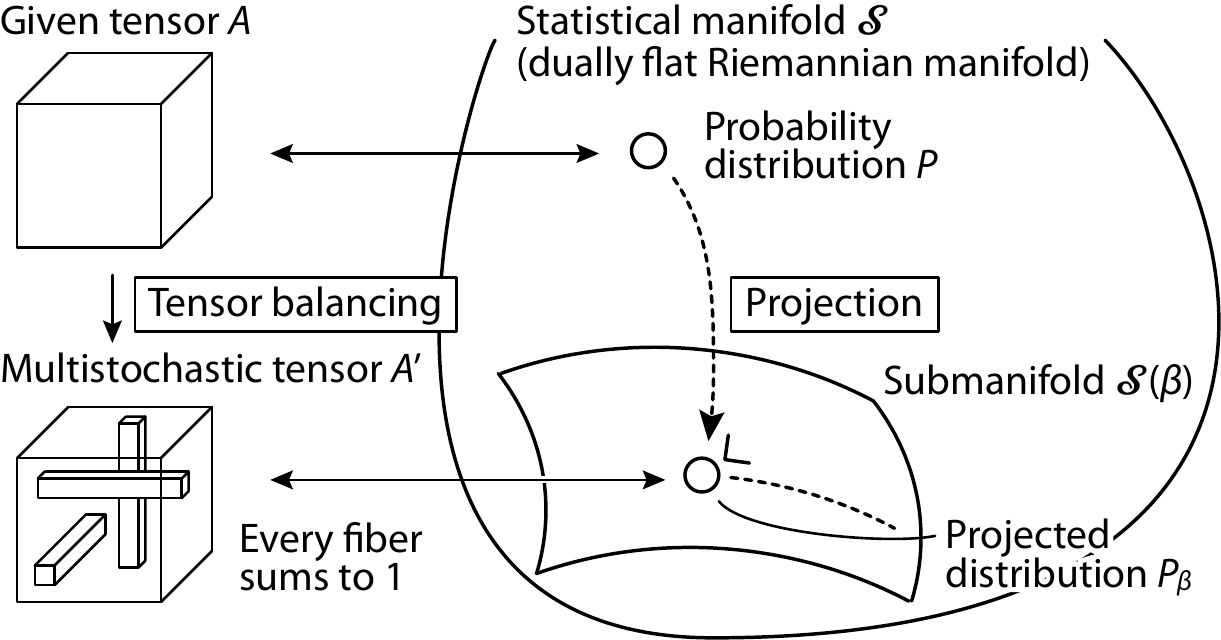}
 \caption{Overview of our approach.}
 \label{fig:overview}
\end{figure}

The remainder of this paper is organized as follows:
We begin with a low-level description of our matrix balancing algorithm in Section~\ref{sec:algorithm} and demonstrate its efficiency in numerical experiments in Section~\ref{sec:exp}.
To guarantee the correctness of the algorithm and extend it to tensor balancing, we provide theoretical analysis in Section~\ref{sec:theory}.
In Section~\ref{subsec:formulation}, we introduce a generalized log-linear model associated with a partial order structured outcome space, followed by introducing the dually flat Riemannian structure in Section~\ref{subsec:duallyflat}.
In Section~\ref{subsec:proj}, we show how to use Newton's method to compute projection of a probability distribution onto a submanifold.
Finally, we formulate the matrix and tensor balancing problem in Section~\ref{sec:balance} and summarize our contributions in Section~\ref{sec:conclusion}.

\section{The Matrix Balancing Algorithm}\label{sec:algorithm}
Given a nonnegative square matrix $A = (a_{ij}) \in \R^{n \times n}_{\ge 0}$,
the task of \emph{matrix balancing} is to find $\vec{r}, \vec{s} \in \R^{n}$ that satisfy
\begin{align}
 \label{eq:matrix_balancing}
 (RAS)\vec{1} = \vec{1},\quad
 (RAS)^T\vec{1} = \vec{1},
\end{align}
where $R = \diag(\vec{r})$ and $S = \diag(\vec{s})$.
The balanced matrix $A' = RAS$ is called \emph{doubly stochastic}, in which each entry $a'_{ij} = a_{ij} r_{i} s_{j}$ and all the rows and columns sum to one.
The most popular algorithm is the Sinkhorn-Knopp algorithm, which repeats updating $\vec{r}$ and $\vec{s}$ as $\vec{r} = 1 / (A\vec{s})$ and $\vec{s} = 1 / (A^T \vec{r})$.
We denote by $[n] = \{1, 2, \dots, n\}$ hereafter.

In our algorithm, instead of directly updating $\vec{r}$ and $\vec{s}$, we update two parameters $\theta$ and $\eta$ defined as
\begin{align}
 \label{eq:theta_eta_mat}
 \log p_{ij} = \sum_{i' \le i}\sum_{j' \le j}\theta_{i'j'},\quad
 \eta_{ij} = \sum_{i' \ge i}\sum_{j' \ge j}p_{i'j'}
\end{align}
for each $i, j \in [n]$, where we normalized entries as $p_{ij} = a_{ij} / \sum_{ij} a_{ij}$ so that $\sum_{ij} p_{ij} = 1$.
We assume for simplicity that each entry is strictly larger than zero. The assumption will be removed in Section~\ref{sec:balance}.

\begin{figure}[t]
 \centering
 \includegraphics[width=.6\linewidth]{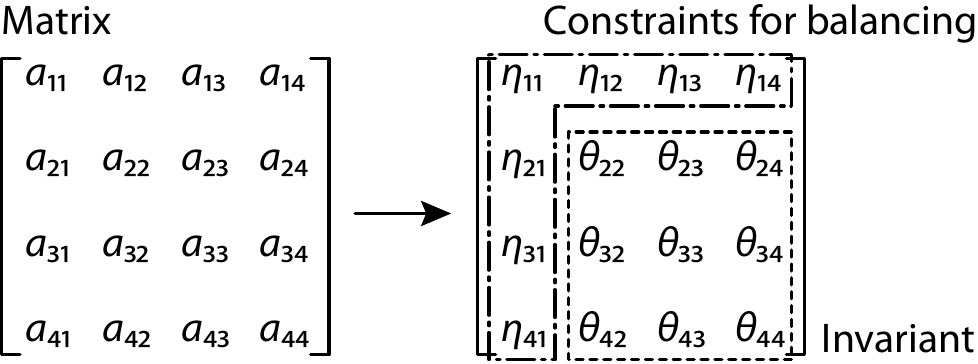}
 \caption{Matrix balancing with two parameters $\theta$ and $\eta$.}
 \label{fig:matrix}
\end{figure}

The key to our approach is that we update $\theta_{ij}^{(t)}$ with $i = 1$ or $j = 1$ by Newton's method at each iteration $t = 1, 2, \dots$ while fixing $\theta_{ij}$ with $i,j \not= 1$ so that $\eta_{ij}^{(t)}$ satisfies the following condition (Figure~\ref{fig:matrix}):
\begin{align*}
 \eta_{i1}^{(t)} = \frac{n - i + 1}{n},\quad
 \eta_{1j}^{(t)} = \frac{n - j + 1}{n}.
\end{align*}
Note that the rows and columns sum not to $1$ but to $1 / n$ due to the normalization.
The update formula is described as
\begin{align}
\label{eq:update_mat}
 \begin{bmatrix}
  \theta_{11}^{(t + 1)}\\
  \theta_{12}^{(t + 1)}\\
  \vdots\\
  \theta_{1n}^{(t + 1)}\\
  \theta_{21}^{(t + 1)}\\
  \vdots\\
  \theta_{n1}^{(t + 1)}
 \end{bmatrix}
 =
 \begin{bmatrix}
  \theta_{11}^{(t)}\\
  \theta_{12}^{(t)}\\
  \vdots\\
  \theta_{1n}^{(t)}\\
  \theta_{21}^{(t)}\\
  \vdots\\
  \theta_{n1}^{(t)}
 \end{bmatrix}
 - J^{-1}
 \begin{bmatrix}
  \eta_{11}^{(t)} - (n - 1 + 1) / n\\
  \eta_{12}^{(t)} - (n - 2 + 1) / n\\
  \vdots\\
  \eta_{1n}^{(t)} - (n - n + 1) / n\\
  \eta_{21}^{(t)} - (n - 2 + 1) / n\\
  \vdots\\
  \eta_{n1}^{(t)} - (n - n + 1) / n
 \end{bmatrix},
\end{align}
where $J$ is the Jacobian matrix given as
\begin{align}
 \label{eq:Jacobian_mat}
 J_{(ij)(i'j')} = \frac{\partial \eta^{(t)}_{ij}}{\partial \theta^{(t)}_{i'j'}} = \eta_{\max\{i, i'\}\max\{j, j'\}} - n^2 \eta_{ij}\eta_{i'j'},
\end{align}
which is derived from our theoretical result in Theorem~\ref{theorem:metric}.
Since $J$ is a $(2n - 1) \times (2n - 1)$ matrix, the time complexity of each update is $O(n^3)$, which is needed to compute the inverse of $J$.

After updating to $\theta_{ij}^{(t + 1)}$, we can compute $p_{ij}^{(t + 1)}$ and $\eta_{ij}^{(t + 1)}$ by Equation~\eqref{eq:theta_eta_mat}.
Since this update does not ensure the condition $\sum_{ij} p_{ij}^{(t + 1)} = 1$, we again update $\theta_{11}^{(t + 1)}$ as
\begin{align*}
 \theta_{11}^{(t + 1)} = \theta_{11}^{(t + 1)} - \log\sum_{ij}p_{ij}^{(t + 1)}
\end{align*}
and recompute $p_{ij}^{(t + 1)}$ and $\eta_{ij}^{(t + 1)}$ for each $i, j \in [n]$.

By iterating the above update process in Equation~\eqref{eq:update_mat} until convergence, $A = (a_{ij})$ with $a_{ij} = n p_{ij}$ becomes doubly stochastic.

\begin{figure}[t]
 \centering
 \includegraphics[width=.6\linewidth]{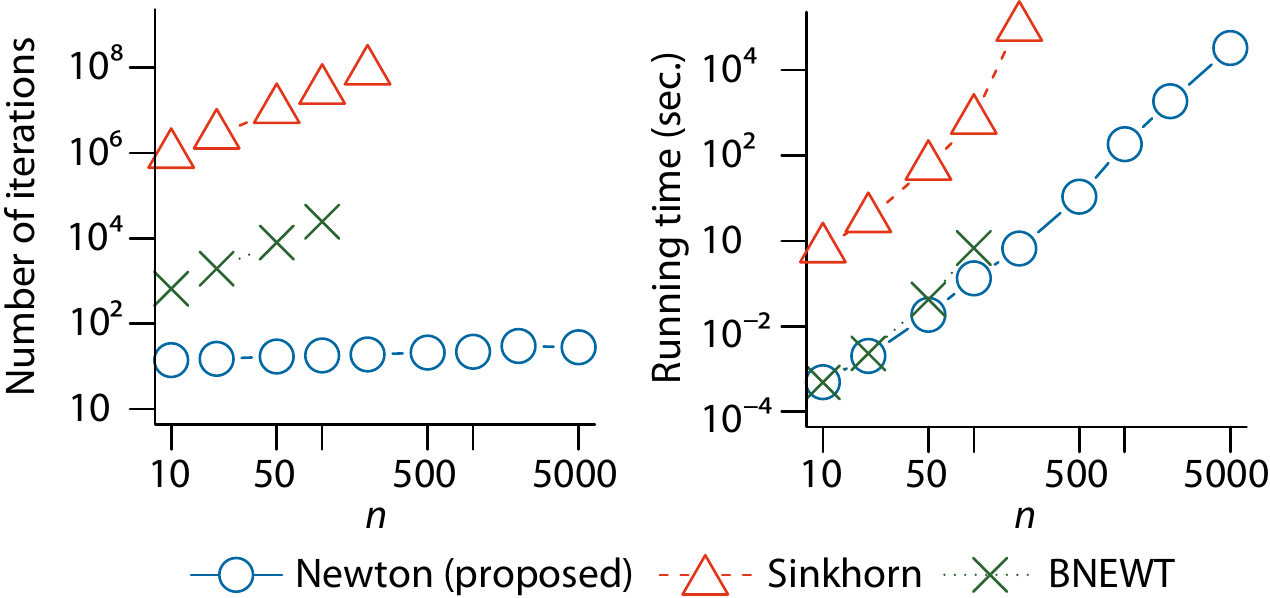}
 \caption{Results on Hessenberg matrices. The BNEWT algorithm (green) failed to converge for $n \ge 200$.}
 \label{fig:Hessenberg}
\end{figure}

\begin{figure}[t]
 \centering
 \includegraphics[width=.6\linewidth]{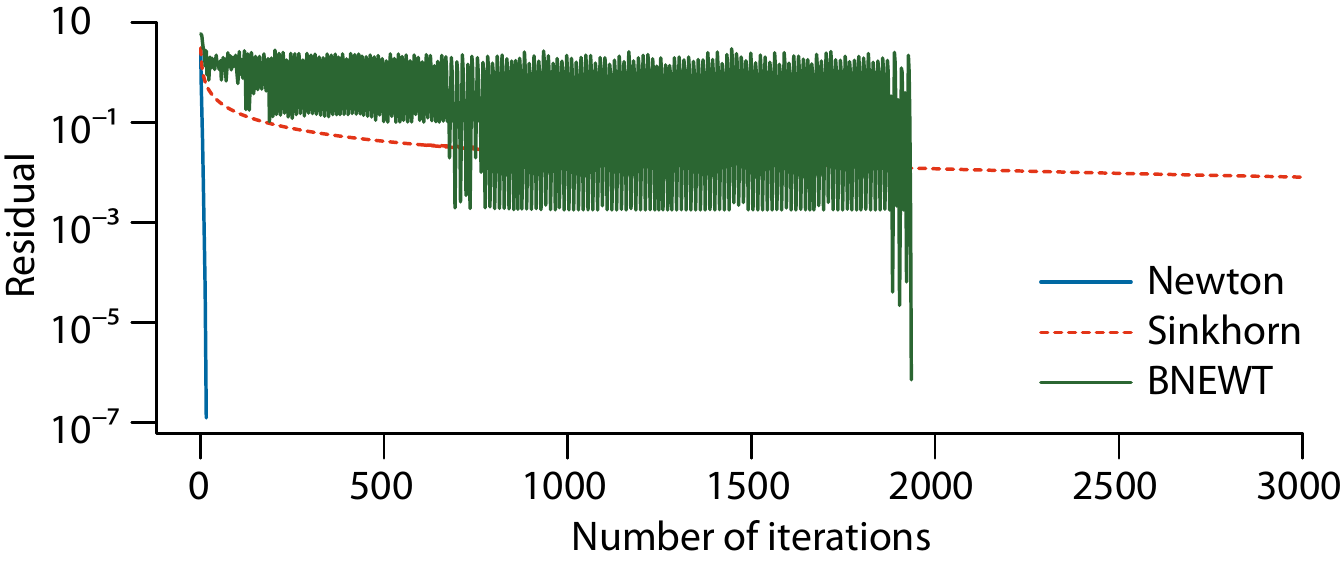}
 \caption{Convergence graph on $H_{20}$.}
 \label{fig:convergence}
\end{figure}

\section{Numerical Experiments}\label{sec:exp}
We evaluate the efficiency of our algorithm compared to the two prominent balancing methods, the standard Sinkhorn-Knopp algorithm~\cite{Sinkhorn64} and the state-of-the-art algorithm BNEWT~\cite{Knight13}, which uses Newton's method-like iterations with conjugate gradients.
All experiments were conducted on Amazon Linux AMI release 2016.09 with a single core of 2.3 GHz Intel Xeon CPU E5-2686 v4 and 256 GB of memory.
All methods were implemented in \texttt{C++} with the \texttt{Eigen} library and compiled with {\ttfamily gcc} 4.8.3\footnote{An implementation of algorithms for matrices and third order tensors is available at: \url{https://github.com/mahito-sugiyama/newton-balancing}}.
We have carefully implemented BNEWT by directly translating the MATLAB code provided in~\cite{Knight13} into \texttt{C++} with the \texttt{Eigen} library for fair comparison, and
used the default parameters.
We measured the residual of a matrix $A' = (a_{ij}')$ by the squared norm $\|(A'\vec{1} - \vec{1}, A'^T\vec{1} - \vec{1})\|_2$, where each entry $a'_{ij}$ is obtained as $n p_{ij}$ in our algorithm, and ran each of three algorithms until the residual is below the tolerance threshold $10^{-6}$.\vspace*{5pt}\\
\textbf{Hessenberg Matrix.}
The first set of experiments used a Hessenberg matrix, which has been a standard benchmark for matrix balancing~\cite{Parlett82,Knight13}.
Each entry of an $n \times n$ Hessenberg matrix $H_n = (h_{ij})$ is given as $h_{ij} = 0$ if $j < i - 1$ and $h_{ij} = 1$ otherwise.
We varied the size $n$ from $10$ to $5,000$, and measured running time (in seconds) and the number of iterations of each method.

Results are plotted in Figure~\ref{fig:Hessenberg}.
Our balancing algorithm with the Newton's method (plotted in blue in the figures) is clearly the fastest:
It is three to five orders of magnitude faster than the standard Sinkhorn-Knopp algorithm (plotted in red).
Although the BNEWT algorithm (plotted in green) is competitive if $n$ is small, it suddenly fails to converge whenever $n \ge 200$, which is consistent with results in the original paper~\cite{Knight13} where there is no result for the setting $n \ge 200$ on the same matrix.
Moreover, our method converges around $10$ to $20$ steps, which is about three and seven orders of magnitude smaller than BNEWT and Sinkhorn-Knopp, respectively, at $n = 100$.

\begin{figure}[t]
 \centering
 \includegraphics[width=.6\linewidth]{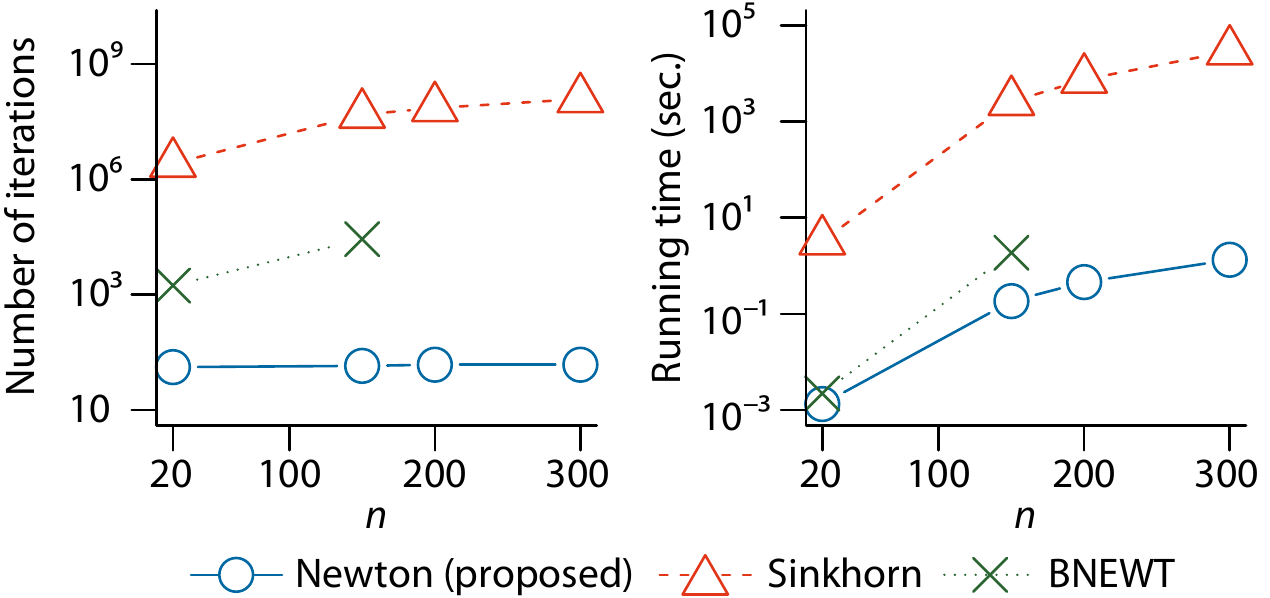}
 \caption{Results on Trefethen matrices. The BNEWT algorithm (green) failed to converge for $n \ge 200$.}
 \label{fig:Trefethen}
\end{figure}

To see the behavior of the rate of convergence in detail, we plot the convergence graph in Figure~\ref{fig:convergence} for $n = 20$, where we observe the slow convergence rate of the Sinkhorn-Knopp algorithm and unstable convergence of the BNEWT algorithm, which contrasts with our quick convergence.\vspace*{5pt}\\
\textbf{Trefethen Matrix.}
Next, we collected a set of Trefethen matrices from a collection website\footnote{\url{http://www.cise.ufl.edu/research/sparse/matrices/}}, which are nonnegative diagonal matrices with primes.
Results are plotted in Figure~\ref{fig:Trefethen}, where we observe the same trend as before:
Our algorithm is the fastest and about four orders of magnitude faster than the Sinkhorn-Knopp algorithm.
Note that larger matrices with $n > 300$ do not have total support, which is the necessary condition for matrix balancing~\cite{Knight13}, while the BNEWT algorithm fails to converge if $n = 200$ or $n = 300$.

\section{Theoretical Analysis}\label{sec:theory}
In the following, we provide theoretical support to our algorithm by formulating the problem as a projection within a statistical manifold, in which a matrix corresponds to an element, that is, a probability distribution, in the manifold.

We show that a balanced matrix forms a submanifold and matrix balancing is projection of a given distribution onto the submanifold, where
the Jacobian matrix in Equation~\eqref{eq:Jacobian_mat} is derived from the gradient of the manifold.

\subsection{Formulation}\label{subsec:formulation}
We introduce our log-linear probabilistic model, where the outcome space is a partially ordered set, or a \emph{poset}~\cite{Gierz03}.
We prepare basic notations and the key mathematical tool for posets, the M{\"o}bius inversion formula, followed by formulating the log-linear model.

\subsubsection{M{\"o}bius Inversion}\label{subsubsec:moebius}
A poset $(S, \le)$, the set of elements $S$ and a partial order $\le$ on $S$, is a fundamental structured space in computer science.
A \emph{partial order} ``$\le$'' is a relation between elements in $S$ that satisfies the following three properties: For all $x, y, z \in S$,
(1) $x \le x$ (reflexivity), (2) $x \le y$, $y \le x \Rightarrow x = y$ (antisymmetry), and (3) $x \le y$, $y \le z \Rightarrow x \le z$ (transitivity).
In what follows, $S$ is always finite and includes the least element (bottom) $\bot \in S$; that is, $\bot \le x$ for all $x \in S$.
We denote $S \setminus \{\bot\}$ by $S^+$.

\citet{Rota64} introduced the \emph{M{\"o}bius inversion formula} on posets by generalizing the inclusion-exclusion principle.
Let $\zeta \from S \times S \to \{0, 1\}$ be the \emph{zeta function} defined as
\begin{align*}
 \zeta(s, x) = \left\{
 \begin{array}{ll}
  1 & \text{if } s \le x, \\
  0 & \text{otherwise}.
 \end{array}
 \right.
\end{align*}
The \emph{M{\"o}bius function} $\mu \from S \times S \to \Z$ satisfies $\zeta\mu = I$, which is inductively defined for all $x, y$ with $x \le y$ as
 \begin{align*}
  \mu(x, y) = \left\{
  \begin{array}{ll}
   1 & \text{if } x = y,\\
   -\sum_{x \le s < y} \mu(x, s) & \text{if } x < y,\\
   0 & \text{otherwise}.
  \end{array}
  \right.
 \end{align*}
From the definition, it follows that
\begin{equation}
\begin{aligned}
 \label{eq:muzero}
 \sum_{s \in S}\zeta(s, y)\mu(x, s) &= \sum_{x \le s \le y}\mu(x, s) = \delta_{xy},\\
 \sum_{s \in S}\zeta(x, s)\mu(s, y) &= \sum_{x \le s \le y}\mu(s, y) = \delta_{xy}
\end{aligned} 
\end{equation}
with the Kronecker delta $\delta$ such that $\delta_{xy} = 1$ if $x = y$ and $\delta_{xy} = 0$ otherwise.
Then for any functions $f$, $g$, and $h$ with the domain $S$ such that
\begin{align*}
 g(x) &= \sum_{s \in S} \zeta(s, x) f(s) = \sum_{s \le x} f(s),\\
 h(x) &= \sum_{s \in S} \zeta(x, s) f(s) = \sum_{s \ge x} f(s),
\end{align*}
$f$ is uniquely recovered with the M{\"o}bius function:
\begin{align*}
 f(x) = \sum_{s \in S} \mu(s, x) g(s),\quad
 f(x) = \sum_{s \in S} \mu(x, s) h(s).
\end{align*}
This is called the \emph{M{\"o}bius inversion formula} and is at the heart of enumerative combinatorics~\cite{Ito93}.

\subsubsection{Log-Linear Model on Posets}\label{subsubsec:model}
We consider a probability vector $p$ on $(S, \le)$ that gives a discrete probability distribution with the outcome space $S$.
A probability vector is treated as a mapping $p\from S\to (0, 1)$ such that $\sum_{x \in S} p(x)= 1$, where every entry $p(x)$ is assumed to be strictly larger than zero.

Using the zeta and the M{\"o}bius functions, let us introduce two mappings $\theta\from S\to \R$ and $\eta\from S \to \R$ as
\begin{align}
 \label{eq:theta}
 \theta(x) &= \sum_{s \in S} \mu(s, x) \log p(s),\\
 \label{eq:eta}
 \eta(x) &= \sum_{s \in S} \zeta(x, s) p(s) = \sum_{s \ge x} p(s).
 \end{align}
From the M{\"o}bius inversion formula, we have
\begin{align}
 \label{eq:logp}
 \log p(x) &= \sum_{s \in S} \zeta(s, x) \theta(s) = \sum_{s \le x} \theta(s),\\
 \label{eq:p}
 p(x) &= \sum_{s \in S} \mu(x, s) \eta(s).
\end{align}
They are generalization of the \emph{log-linear model}~\cite{Agresti12} that gives the probability $p(\vec{x})$ of an $n$-dimensional binary vector $\vec{x} = (x^1, \dots, x^n) \in \{0, 1\}^n$ as
\begin{align*}
 \log p(\vec{x}) = \sum_{i} \theta^i x^i + \sum_{i < j} \theta^{ij} x^i x^j + \sum_{i < j < k} \theta^{ijk} x^i x^j x^k + \dots + \theta^{1\dots n} x^1 x^2 \dots x^n - \psi,
\end{align*}
where $\vec{\theta} = (\theta^{1}, \dots, \theta^{12\dots n})$ is a parameter vector, $\psi$ is a normalizer, and $\vec{\eta} = (\eta^{1}, \dots, \eta^{12\dots n})$ represents the expectation of variable combinations such that
\begin{align*}
 \eta^{i} &= \Exp[x^i] = \Pr(x^i = 1),\\
 \eta^{ij} &= \Exp[x^ix^j] = \Pr(x^i = x^j = 1),\ i < j, \dots\\
 \eta^{1\dots n} &= \Exp[x^1 \dots x^n] = \Pr(x^1 = \dots = x^n = 1).
\end{align*}
They coincide with Equations~\eqref{eq:logp} and~\eqref{eq:eta} when we let $S = 2^V$ with $V = \{1, 2, \dots, n\}$, each $x \in S$ as the set of indices of ``$1$'' of $\vec{x}$, and the order $\le$ as the inclusion relationship, that is, $x \le y$ if and only if $x \subseteq y$.
\citet{Nakahara06} have pointed out that $\vec{\theta}$ can be computed from $p$ using the inclusion-exclusion principle in the log-linear model.
We exploit this combinatorial property of the log-linear model using the M{\"o}bius inversion formula on posets and extend the log-linear model from the power set $2^V$ to any kind of posets $(S, \le)$.
\citet{Sugiyama2016ISIT} studied a relevant log-linear model, but the relationship with M{\"o}bius inversion formula has not been analyzed yet.

\subsection{Dually Flat Riemannian Manifold}\label{subsec:duallyflat}
We theoretically analyze our log-linear model introduced in Equations~\eqref{eq:theta},~\eqref{eq:eta} and show that they form dual coordinate systems on a dually flat manifold, which has been mainly studied in the area of information geometry~\cite{Amari01,Nakahara02,Amari14,Amari16}.
Moreover, we show that the Riemannian metric and connection of our model can be analytically computed in closed forms.

In the following, we denote by $\xi$ the function $\theta$ or $\eta$ and by $\nabla$ the gradient operator with respect to $S^+ = S \setminus \{\bot\}$, i.e., $(\nabla f(\xi))(x) = \partial f / \partial \xi(x)$ for $x \in S^+$, and denote by $\vec{\Sc}$ the set of probability distributions specified by probability vectors, which forms a statistical manifold.
We use uppercase letters $P, Q, R, \dots$ for points (distributions) in $\vec{\Sc}$ and their lowercase letters $p, q, r, \dots$ for the corresponding probability vectors treated as mappings.
We write $\theta_P$ and $\eta_P$ if they are connected with $p$ by Equations~\eqref{eq:theta} and~\eqref{eq:eta}, respectively, and abbreviate subscripts if there is no ambiguity.

\subsubsection{Dually Flat Structure}\label{subsubsec:dual}
We show that $\vec{\Sc}$ has the \emph{dually flat Riemannian structure} induced by two functions $\theta$ and $\eta$ in Equation~\eqref{eq:theta} and~\eqref{eq:eta}.
We define $\psi(\theta)$ as
\begin{align}
 \label{eq:psi}
 \psi(\theta) = -\theta(\bot) = -\log p(\bot),
\end{align}
which corresponds to the normalizer of $p$.
It is a convex function since we have
\begin{align*}
 \psi(\theta) = \log \sum_{x \in S} \exp\left(\,\sum_{\bot < s \le x} \theta(s)\,\right)
\end{align*}
from $\log p(x) = \sum_{\bot < s \le x} \theta(s) - \psi(\theta)$.
We apply the \emph{Legendre transformation} to $\psi(\theta)$ given as
\begin{align}
 \label{eq:phi}
 \phi(\eta) = \max_{\theta'} \Big(\theta'\eta - \psi(\theta')\Big),\quad
 \theta' \eta = \sum_{\mathclap{x \in S^+}} \theta'(x) \eta(x).
\end{align}
Then $\phi(\eta)$ coincides with the negative entropy.
\begin{theorem}[Legendre dual]
 \label{theorem:negentropy}
 \begin{align*}
   \phi(\eta) = \sum_{x \in S} p(x) \log p(x).
 \end{align*}
\end{theorem}
\begin{proof}
 From Equation~\eqref{eq:muzero}, we have
 \begin{align*}
  \theta' \eta
  = \sum_{x \in S^+} \left(\,\sum_{\bot < s \le x} \mu(s, x) \log p'(s) \sum_{s \ge x} p(s)\,\right)
  = \sum_{x \in S^+} p(x)\left(\,\log p'(x) - \log p'(\bot)\,\right).
 \end{align*}
 Thus it holds that
 \begin{align}
  \label{eq:relentropy}
  \theta'\eta - \psi(\theta')
  &= \sum_{x \in S} p(x) \log p'(x).
 \end{align}
 Hence it is maximized with $p(x) = p'(x)$.
\end{proof}
\noindent Since they are connected with each other by the Legendre transformation, they form a \emph{dual coordinate system} $\nabla\psi(\theta)$ and $\nabla\phi(\eta)$ of $\vec{\Sc}$~\citep[Section~1.5]{Amari16}, which coincides with $\theta$ and $\eta$ as follows.
\begin{theorem}[dual coordinate system]
 \begin{align}
  \label{eq:psiphi}
  \nabla \psi(\theta) = \eta,\quad
  \nabla \phi(\eta) = \theta.
 \end{align}
\end{theorem}
\begin{proof}
 They can be directly derived from our definitions (Equations~\eqref{eq:theta} and~\eqref{eq:phi}) as
\begin{align*}
 \frac{\partial\psi(\theta)}{\partial \theta(x)}
 &= \frac{\sum_{y \ge x} \exp\left(\sum_{\bot < s \le y} \theta(s)\right)}{\sum_{y \in S} \exp\left(\sum_{\bot < s \le y} \theta(s)\right)}
 = \sum_{s \ge x} p(s) = \eta(x),\\
 \frac{\partial\phi(\eta)}{\partial \eta(x)}
 &= \frac{\partial}{\partial \eta(x)} \Big(\theta\eta - \psi(\theta)\Big) = \theta(x).&&\qedhere
\end{align*}
\end{proof}
Moreover, we can confirm the \emph{orthogonality} of $\theta$ and $\eta$ as
\begin{align*}
 \Exp\left[\frac{\partial}{\partial \theta(x)}\log p(s) \frac{\partial}{\partial \eta(y)}\log p(s)\right]
 &= \sum_{s \in S}\left[\,p(s) \frac{\partial}{\partial \theta(x)} \sum_{u \in S} \zeta(u, s)\theta(u) \frac{\partial}{\partial \eta(y)} \log\left(\sum_{u \in S} \mu(s, u)\eta(u)\right)\,\right]\\
 &= \sum_{s \in S}\left[\,p(s) \left(\zeta(x, s) - \eta(x)\right) \frac{\mu(s, y)}{p(s)}\,\right]\\
 &= \sum_{s \in S} \zeta(x, s)\mu(s, y) = \delta_{xy}.
\end{align*}
The last equation holds from Equation~\eqref{eq:muzero}, hence the M{\"o}bius inversion directly leads to the orthogonality.

The \emph{Bregman divergence} is known to be the canonical divergence~\citep[Section~6.6]{Amari16} to measure the difference between two distributions $P$ and $Q$ on a dually flat manifold, which is defined as
\begin{align*}
 D\left[P, Q\right] = \psi(\theta_P) + \phi(\eta_Q) - \theta_P \eta_Q.
\end{align*}
In our case, since we have
$\phi(\eta_Q) = \sum_{x \in S} q(x) \log q(x)$ and
$\theta_P\eta_Q - \psi(\theta_P) = \sum_{x \in S} q(x) \log p(x)$
from Theorem~\ref{theorem:negentropy} and Equation~\eqref{eq:relentropy},
it is given as
\begin{align*}
 D\left[P, Q\right]
 = \sum_{x \in S} q(x) \log\frac{q(x)}{p(x)},
\end{align*}
which coincides with the \emph{Kullback--Leibler divergence} (KL divergence) from $Q$ to $P$: $D\left[P, Q\right] = \KL\left[Q, P\right]$.

\subsubsection{Riemannian Structure}\label{subsubsec:riemann}
Next we analyze the Riemannian structure on $\vec{\Sc}$ and show that the M{\"o}bius inversion formula enables us to compute the Riemannian metric of $\vec{\Sc}$.
\begin{theorem}[Riemannian metric]
 \label{theorem:metric}
 The manifold $(\vec{\Sc}, g(\xi))$ is a Riemannian manifold with the Riemannian metric $g(\xi)$ such that for all $x, y \in S^+$
 \begin{align*}
 g_{xy}(\xi) = \left\{
 \begin{array}{*2{>{\displaystyle}l}}
  \sum_{s \in S} \zeta(x, s)\zeta(y, s)p(s) - \eta(x)\eta(y) & \text{if } \xi = \theta,\vspace*{5pt}\\
  \sum_{s \in S} \mu(s, x)\mu(s, y) p(s)^{-1} & \text{if } \xi = \eta.
 \end{array}
 \right.
 \end{align*}
\end{theorem}
\begin{proof}
 Since the Riemannian metric is defined as
\begin{align*}
 g(\theta) = \nabla\nabla\psi(\theta),\quad g(\eta) = \nabla\nabla\phi(\eta),
\end{align*}
when $\xi = \theta$ we have
\begin{align*}
 g_{xy}(\theta) &= \frac{\partial^2}{\partial \theta(x) \partial \theta(y)} \psi(\theta) = \frac{\partial}{\partial \theta(x)} \eta(y)
 = \frac{\partial}{\partial \theta(x)} \sum_{s \in S} \zeta(y, s)\exp\left(\sum_{\bot < u \le s} \theta(u) - \psi(\theta)\right)\\
 &= \sum_{s \in S} \zeta(x, s)\zeta(y, s)p(s) - \eta(x)\eta(y).
\end{align*}
When $\xi = \eta$, it follows that
\begin{align*}
 g_{x, y}(\eta) &= \frac{\partial^2}{\partial \eta(x) \partial \eta(y)} \phi(\eta) = \frac{\partial}{\partial \eta(x)} \theta(y)\\
 &= \frac{\partial}{\partial \eta(x)} \sum_{s \le y} \mu(s, y) \log p(s)
 = \frac{\partial}{\partial \eta(x)} \sum_{s \le y} \mu(s, y) \log \left(\sum_{u \ge s} \mu(s, u) \eta(u)\right)\\
 &= \sum_{s \in S} \frac{\mu(s, x)\mu(s, y)}{\sum_{u \ge s} \mu(s, u)\eta(u)}
 = \sum_{s \in S} \mu(s, x)\mu(s, y) p(s)^{-1}.\qedhere
\end{align*}
\end{proof}
\noindent Since $g(\xi)$ coincides with the Fisher information matrix,
\begin{align*}
 \Exp\left[\,\frac{\partial}{\partial \theta(x)} \log p(s) \frac{\partial}{\partial \theta(y)} \log p(s)\,\right] &= g_{xy}(\theta),\\
 \Exp\left[\,\frac{\partial}{\partial \eta(x)} \log p(s) \frac{\partial}{\partial \eta(y)} \log p(s)\,\right] &= g_{xy}(\eta).
\end{align*}
Then the Riemannian (Levi--Chivita) connection $\Gamma(\xi)$ with respect to $\xi$, which is defined as
\begin{align*}
 \Gamma_{xyz}(\xi) = \frac{1}{2}\left(\, \frac{\partial g_{yz}(\xi)}{\partial \xi(x)} + \frac{\partial g_{xz}(\xi)}{\partial \xi(y)} - \frac{\partial g_{xy}(\xi)}{\partial \xi(z)}\,\right)
\end{align*}
for all $x, y, z \in S^+$, can be analytically obtained.
\begin{theorem}[Riemannian connection]
 The Riemannian connection $\Gamma(\xi)$ on the manifold $(\vec{\Sc}, g(\xi))$ is given in the following for all $x, y, z \in S^+$,
 \begin{align*}
  \Gamma_{xyz}(\xi) = \left\{
  \begin{array}{*2{>{\displaystyle}l}}
   \phantom{-}\frac{1}{2}\sum_{s \in S}\phantom{}\big(\zeta(x, s) - \eta(x)\big)\big(\zeta(y, s) - \eta(y)\big)\big(\zeta(z, s) - \eta(z)\big) p(s) & \text{if } \xi = \theta,\\
   -\frac{1}{2}\sum_{s \in S} \mu(s, x)\mu(s, y)\mu(s, z) p(s)^{-2} & \text{if } \xi = \eta.
  \end{array}
  \right.
 \end{align*}
\end{theorem}
\begin{proof}
 We have for all $x, y, z \in S$,
 \begin{align*}
  \frac{\partial g_{y,z}(\theta)}{\partial \theta(x)}
  = \frac{\partial}{\partial \theta(x)}\sum_{s \in S}\zeta(y, s)\zeta(z, s)p(s) - \frac{\partial}{\partial \theta(x)}\eta(y)\eta(z),
 \end{align*}
 where
 \begin{align*}
  \frac{\partial}{\partial \theta(x)}\sum_{s \in S}\zeta(y, s)\zeta(z, s)p(s)
  &= \frac{\partial}{\partial \theta(x)} \sum_{s \in S}\zeta(x, s)\zeta(y, s)\zeta(z, s) \exp\left(\sum_{\bot < u \le s} \theta(u) - \psi(\theta)\right)\\
  &= \sum_{s \in S}\zeta(x, s)\zeta(y, s)\zeta(z, s)p(s) - \eta(x)\sum_{s \in S}\zeta(y, s)\zeta(z, s)p(s)
 \end{align*}
 and
 \begin{align*}
  \frac{\partial}{\partial \theta(x)}\eta(y)\eta(z)
  &= \frac{\partial \eta(y)}{\partial \theta(x)}\eta(z) + \frac{\partial \eta(z)}{\partial \theta(x)}\eta(y)\\
  &= \eta(z)\sum_{s \in S}\zeta(x, s)\zeta(y, s)p(s) + \eta(y)\sum_{s \in S}\zeta(x, s)\zeta(z, s)p(s) - 2\eta(x)\eta(y)\eta(z).
 \end{align*}
 It follows that
 \begin{align*}
  \frac{\partial g_{y,z}(\theta)}{\partial \theta(x)}
  = \sum_{s \in S} \big(\zeta(x, s) - \eta(x)\big)\big(\zeta(y, s) - \eta(y)\big)\big(\zeta(z, s) - \eta(z)\big) p(s).
 \end{align*}
 On the other hand,
 \begin{align*}
  \frac{\partial g_{y,z}(\eta)}{\partial \eta(x)}
  &= \frac{\partial}{\partial \eta(x)} \sum_{s \in S} \mu(s, y)\mu(s, z)p(s)^{-1}
  = \frac{\partial}{\partial \eta(x)} \sum_{s \in S} \mu(s, y)\mu(s, z) \left(\sum_{u \ge s} \mu(s, u) \eta(s)\right)^{-1}\\
  &= -\sum_{s \in S} \mu(s, x)\mu(s, y)\mu(s, z) \left(\sum_{u \ge s} \mu(s, u) \eta(s)\right)^{-2}
  = -\sum_{s \in S} \mu(s, x)\mu(s, y)\mu(s, z) p(s)^{-2}.
 \end{align*}
Therefore, from the definition of $\Gamma(\xi)$, it follows that
 \begin{align*}
  \Gamma_{x,y,z}(\theta)
  &= \frac{1}{2}\left(\, \frac{\partial g_{y,z}(\theta)}{\partial \theta(x)} + \frac{\partial g_{x,z}(\theta)}{\partial \theta(y)} - \frac{\partial g_{x,y}(\theta)}{\partial \theta(z)}\,\right)\\
  &= \frac{1}{2}\sum_{s \in S} \big(\zeta(s, x) - \eta(x)\big)\big(\zeta(s, y) - \eta(y)\big)\big(\zeta(s, z) - \eta(z)\big) p(s),\\
  \Gamma_{x,y,z}(\eta)
  &= \frac{1}{2}\left(\, \frac{\partial g_{y,z}(\eta)}{\partial \eta(x)} + \frac{\partial g_{x,z}(\eta)}{\partial \eta(y)} - \frac{\partial g_{x,y}(\eta)}{\partial \eta(z)}\,\right)\\
  &= -\frac{1}{2}\sum_{s \in S} \mu(s, x)\mu(s, y)\mu(s, z) p(s)^{-2}.\qedhere
 \end{align*}
 \end{proof}

\subsection{The Projection Algorithm}\label{subsec:proj}
Projection of a distribution onto a submanifold is essential; several machine learning algorithms are known to be formulated as projection of a distribution empirically estimated from data onto a submanifold that is specified by the target model~\cite{Amari16}.
Here we define projection of distributions on posets and show that Newton's method can be applied to perform projection as the Jacobian matrix can be analytically computed.

\subsubsection{Definition}\label{subsubec:decomp}
Let $\vec{\Sc}(\beta)$ be a submanifold of $\vec{\Sc}$ such that
\begin{align}
 \label{eq:submanifold}
 \vec{\Sc}(\beta) = \Set{P \in \vec{\Sc} | \theta_P(x) = \beta(x) \text{ for all } x \in \dom(\beta)}
\end{align}
specified by a function $\beta$ with $\dom(\beta) \subseteq S^+$.
Projection of $P \in \vec{\Sc}$ onto $\vec{\Sc}(\beta)$, called \emph{$m$-projection}, which is defined as the distribution $P_{\beta} \in \vec{\Sc}(\beta)$ such that
\begin{align*}
 \left\{
 \begin{array}{ll}
  \theta_{P_{\beta}}(x) = \beta(x) & \text{if } x \in \dom(\beta), \\
  \eta_{P_{\beta}}(x) = \eta_P(x) & \text{if } x \in S^+ \setminus \dom(\beta),
 \end{array}
 \right.
\end{align*}
is the minimizer of the KL divergence from $P$ to $\vec{\Sc}(\beta)$:
\begin{align*}
 P_\beta &= \argmin_{Q \in \vec{\Sc}(\beta)} \KL[P, Q].
\end{align*}
The dually flat structure with the coordinate systems $\theta$ and $\eta$ guarantees that the projected distribution $P_{\beta}$ always exists and is unique~\citep[Theorem~3]{Amari09}.
Moreover, the \emph{Pythagorean theorem} holds in the dually flat manifold, that is, for any $Q \in \vec{\Sc}(\beta)$ we have
\begin{align*}
 \KL[P, Q] = \KL[P, P_{\beta}] + \KL[P_{\beta}, Q].
\end{align*}
We can switch $\eta$ and $\theta$ in the submanifold $\vec{\Sc}(\beta)$ by changing $\KL[P, Q]$ to $\KL[Q, P]$, where the projected distribution $P_{\beta}$ of $P$ is given as
\begin{align*}
 \left\{
 \begin{array}{ll}
  \theta_{P_{\beta}}(x) = \theta_P(x) & \text{if } x \in S^+ \setminus \dom(\beta), \\
  \eta_{P_{\beta}}(x) = \beta(x) & \text{if } x \in \dom(\beta),
 \end{array}
 \right.
\end{align*}
This projection is called \emph{$e$-projection}.

\begin{example}[Boltzmann machine]
 Given a Boltzmann machine represented as an undirected graph $G = (V, E)$ with a vertex set $V$ and an edge set $E \subseteq \{ \{i, j\} \mid i, j \in V\}$.
 The set of probability distributions that can be modeled by a Boltzmann machine $G$ coincides with the submanifold
 \begin{align*}
  \vec{\Sc}_{\mathrm{B}} = \Set{P \in \vec{\Sc} | \theta_P(x) = 0 \text{ if } |x| > 2 \text{ or } x \not\in E},
 \end{align*}
 with $S = 2^V$.
 Let $\hat{P}$ be an empirical distribution estimated from a given dataset.
 The learned model is the $m$-projection of the empirical distribution $\hat{P}$ onto $\vec{\Sc}_{\mathrm{B}}$, where the resulting distribution $P_{\beta}$ is given as
 \begin{align*}
  \left\{
  \begin{array}{ll}
   \theta_{P_{\beta}}(x) = 0 & \text{if } |x| > 2 \text{ or } x \not\in E, \\
   \eta_{P_{\beta}}(x) = \eta_{\hat{P}}(x) & \text{if } |x| = 1 \text{ or } x \in E.
  \end{array}
  \right.
 \end{align*}
\end{example}

\subsubsection{Computation}\label{subsubsec:newton}
Here we show how to compute projection of a given probability distribution.
We show that Newton's method can be used to efficiently compute the projected distribution $P_{\beta}$ by iteratively updating $P_{\beta}^{(0)} = P$ as $P_{\beta}^{(0)}, P_{\beta}^{(1)}, P_{\beta}^{(2)}, \dots$ until converging to $P_{\beta}$.

Let us start with the $m$-projection with initializing $P_{\beta}^{(0)} = P$.
In each iteration $t$, we update $\theta_{P_{\beta}}^{(t)}(x)$ for all $x \in \dom{\beta}$ while fixing $\eta_{P_{\beta}}^{(t)}(x) = \eta_P(x)$ for all $x \in S^+ \setminus \dom(\beta)$, which is possible from the orthogonality of $\theta$ and $\eta$.
Using Newton's method, $\eta_{P_\beta}^{(t + 1)}(x)$ should satisfy
\begin{align*}
 \left(\theta_{P_{\beta}}^{(t)}(x) - \beta(x)\right) + \sum_{\mathclap{y \in \dom(\beta)}} J_{xy} \left(\eta_{P_{\beta}}^{(t + 1)}(y) - \eta_{P_{\beta}}^{(t)}(y)\right) = 0,
\end{align*}
for every $x \in \dom(\beta)$, where $J_{xy}$ is an entry of the $|\dom(\beta)| \times |\dom(\beta)|$ Jacobian matrix $J$ and given as
\begin{align*}
 J_{xy} = \frac{\partial \theta_{P_{\beta}}^{(t)}(x)}{\partial \eta_{P_{\beta}}^{(t)}(y)} = \sum_{s \in S} \mu(s, x)\mu(s, y) p_{\beta}^{(t)}(s)^{-1}
\end{align*}
from Theorem~\ref{theorem:metric}.
Therefore, we have the update formula for all $x \in \dom(\beta)$ as
\begin{align*}
 \eta_{P_{\beta}}^{(t + 1)}(x) &= \eta_{P_{\beta}}^{(t)}(x) - \sum_{\mathclap{y \in \dom(\beta)}} J_{xy}^{-1}\left(\theta_{P_{\beta}}^{(t)}(y) - \beta(y)\right).
\end{align*}
In $e$-projection, update $\eta_{P_\beta}^{(t)}(x)$ for $x \in \dom(\beta)$ while fixing $\theta_{P_\beta}^{(t)}(x) = \theta_P(x)$ for all $x \in S^+ \setminus \dom(\beta)$.
To ensure $\eta_{P_\beta}^{(t)}(\bot) = 1$, we add $\bot$ to $\dom(\beta)$ and $\beta(\bot) = 1$.
We update $\theta_{P_\beta}^{(t)}(x)$ at each step $t$ as
\begin{align*}
 \theta_{P_{\beta}}^{(t + 1)}(x) &= \theta_{P_{\beta}}^{(t)}(x) - \sum_{\mathclap{y \in \dom(\beta)}} {J'}_{xy}^{-1}\left(\eta_{P_{\beta}}^{(t)}(y) - \beta(y)\right),\\
 {J'}_{xy} &= \frac{\partial \eta_{P_{\beta}}^{(t)}(x)}{\partial\theta_{P_{\beta}}^{(t)}(y)} = \sum_{s \in S} \zeta(x, s)\zeta(y, s)p_{\beta}^{(t)}(s) - |S|\eta_{P_{\beta}}^{(t)}(x)\eta_{P_{\beta}}^{(t)}(y).
\end{align*}
In this case, we also need to update $\theta_{P_{\beta}}^{(t)}(\bot)$ as it is not guaranteed to be fixed.
Let us define
\begin{align*}
 p_{\beta}'^{(t + 1)}(x) &= p_{\beta}^{(t)}(x)\prod_{\mathclap{s \in \dom(\beta)}}\hspace*{5pt} \frac{\exp\left(\theta_{P_{\beta}}^{(t + 1)}(s)\right)}{\exp\left(\theta_{P_{\beta}}^{(t)}(s)\right)} \,\zeta(s, x).
\end{align*}
Since we have
\begin{align*}
 p_{\beta}^{(t + 1)}(x) &= \frac{\exp\left(\theta_{P_{\beta}}^{(t + 1)}(\bot)\right)}{\exp\left(\theta_{P_{\beta}}^{(t)}(\bot)\right)} p_{\beta}'^{(t + 1)}(x),
\end{align*}
it follows that
\begin{align*}
 \theta_{P_{\beta}}^{(t + 1)}(\bot) - \theta_{P_{\beta}}^{(t)}(\bot) = - \log\left(\exp\left(\theta_{P_{\beta}}^{(t)}(\bot)\right) +\sum_{x \in S^+}p_{\beta}'^{(t + 1)}(x)\right),
\end{align*}
The time complexity of each iteration is $O(|\dom(\beta)|^3)$, which is required to compute the inverse of the Jacobian matrix.

Global convergence of the projection algorithm is always guaranteed by the convexity of a submanifold $\vec{\Sc}(\beta)$ defined in Equation~\eqref{eq:submanifold}.
Since $\vec{\Sc}(\beta)$ is always convex with respect to the $\theta$- and $\eta$-coordinates, it is straightforward to see that our $e$-projection is an instance of the \emph{Bregman algorithm} onto a convex region, which is well known to always converge to the global solution~\cite{Censor81}.

\section{Balancing Matrices and Tensors}\label{sec:balance}
Now we are ready to solve the problem of matrix and tensor balancing as projection on a dually flat manifold.

\subsection{Matrix Balancing}\label{subsec:matrix}
Recall that the task of matrix balancing is to find $\vec{r}, \vec{s} \in \R^{n}$ that satisfy $(RAS)\vec{1} = \vec{1}$ and $(RAS)^T\vec{1} = \vec{1}$ with $R = \diag(\vec{r})$ and $S = \diag(\vec{s})$
for a given nonnegative square matrix $A = (a_{ij}) \in \R^{n \times n}_{\ge 0}$.

Let us define $S$ as
\begin{align}
 \label{eq:matrix_S}
 S = \Set{(i, j) | i, j \in [n] \text{ and } a_{ij} \not= 0},
\end{align}
where we remove zero entries from the outcome space $S$ as our formulation cannot treat zero probability, and give each probability as $p((i, j)) = a_{ij} / \sum_{ij} a_{ij}$.
The partial order $\le$ of $S$ is naturally introduced as
\begin{align}
 \label{eq:matrix_order}
 x = (i, j) \le y = (k, l) \Leftrightarrow i \le j \text{ and } k \le l,
\end{align}
resulting in $\bot = (1, 1)$.
In addition, we define $\vec{\iota}_{k, m}$ for each $k \in [n]$ and $m \in \{1, 2\}$ such that
\begin{align*}
 \vec{\iota}_{k, m} = \min\Set{x = (i_1, i_2) \in S | i_m = k},
\end{align*}
where the minimum is with respect to the order $\le$.
If $\vec{\iota}_{k, m}$ does not exist, we just remove the entire $k$th row if $m = 1$ or $k$th column if $m = 2$ from $A$.
Then we switch rows and columns of $A$ so that the condition
\begin{align}
 \label{eq:index_condition}
 \vec{\iota}_{1, m} \le \vec{\iota}_{2, m} \le \dots \le \vec{\iota}_{n, m}
\end{align}
is satisfied for each $m \in \{1, 2\}$, which is possible for any matrices. 
Since we have
\begin{align*}
 \eta(\vec{\iota}_{k, m}) - \eta(\vec{\iota}_{k + 1, m}) = \left\{
 \begin{array}{ll}
  \sum_{j = 1}^{n} p((k, j)) & \text{if } m = 1, \\[3pt]
  \sum_{i = 1}^{n} p((i, k)) & \text{if } m = 2
 \end{array}
 \right.
\end{align*}
if the condition~\eqref{eq:index_condition} is satisfied, the probability distribution is balanced if for all $k \in [n]$ and $m \in \{1, 2\}$
\begin{align*}
 \eta(\vec{\iota}_{k, m}) = \frac{n \!-\! k \!+\! 1}{n}.
\end{align*}
Therefore, we obtain the following result.

\paragraph{Matrix balancing as \textit{e}-projection:}
Given a matrix $A \in \R^{n \times n}$ with its normalized probability distribution $P \in \vec{\Sc}$ such that $p((i, j)) = a_{ij} / \sum_{ij} a_{ij}$.
Define the poset $(S, \le)$ by Equations~\eqref{eq:matrix_S} and~\eqref{eq:matrix_order} and let $\vec{\Sc}(\beta)$ be the submanifold of $\vec{\Sc}$ such that
\begin{align*}
 \vec{\Sc}(\beta) = \Set{P \in \vec{\Sc} | \eta_P(x) = \beta(x) \text{ for all } x \in \dom(\beta)},
\end{align*}
where the function $\beta$ is given as
\begin{align*}
 \dom(\beta) &= \Set{\vec{\iota}_{k, m} \in S | k \in [n], m \in \{1, 2\}},\\
 \beta(\vec{\iota}_{k, m}) &= \frac{n \!-\! k \!+\! 1}{n}.
\end{align*}
Matrix balancing is the $e$-projection of $P$ onto the submanifold $\vec{\Sc}(\beta)$, that is, the balanced matrix $(RAS) / n$ is the distribution $P_{\beta}$ such that
\begin{align*}
 \left\{
 \begin{array}{ll}
  \theta_{P_{\beta}}(x) = \theta_P(x) & \text{if } x \in S^+ \setminus \dom(\beta),\\
  \eta_{P_{\beta}}(x) = \beta(x) & \text{if } x \in \dom(\beta),
 \end{array}
 \right.
\end{align*}
which is unique and always exists in $\vec{\Sc}$, thanks to its dually flat structure.
Moreover, two balancing vectors $\vec{r}$ and $\vec{s}$ are
\begin{align*}
 \exp\left(\,\sum_{k = 1}^i \theta_{P_\beta}(\vec{\iota}_{k, m}) - \theta_{P}(\vec{\iota}_{k, m})\right) = \left\{
 \begin{array}{ll}
  r_i & \text{if } m = 1,\\
  a_i & \text{if } m = 2,
 \end{array}
 \right.
\end{align*}
for every $i \in [n]$ and $\vec{r} = \vec{r} n / \sum_{ij} a_{ij}$.\hfill$\blacksquare$

\subsection{Tensor Balancing}\label{subsec:tensor}
Next, we generalize our approach from matrices to tensors.
For an $N$th order tensor $A = (a_{i_1 i_2 \dots i_N}) \in \R^{n_1 \times n_2 \times \dots \times n_N}$ and a vector $\vec{b} \in \R^{n_m}$, the $m$-mode product of $A$ and $\vec{b}$ is defined as
\begin{align*}
 \left(A \times_m \vec{b}\right)_{i_1 \dots i_{m - 1} i_{m + 1} \dots i_N} = \sum_{i_m = 1}^{n_m} a_{i_1 i_2 \dots i_N} b_{i_m}.
\end{align*}
We define \emph{tensor balancing} as follows:
Given a tensor $A \in \R^{n_1 \times n_2 \times \dots \times n_N}$ with $n_1 = \dots = n_N = n$,
find $(N - 1)$ order tensors $R^1, R^2, \dots, R^N$ such that
\begin{align}
 \label{eq:tensor_balancing}
 A' \times_{m} \vec{1} = \vec{1}\quad (\in \R^{n_1 \times \dots \times n_{m - 1} \times n_{m + 1} \times \dots \times n_{N}})
\end{align}
for all $m \in [N]$, i.e., $\sum_{i_m = 1}^{n} a_{i_1 i_2\dots i_N}' = 1$, where each entry $a'_{i_1 i_2\dots i_N}$ of the balanced tensor $A'$ is given as
\begin{align*}
 a'_{i_1 i_2\dots i_N} = a_{i_1 i_2 \dots i_N}\ \prod_{\mathclap{m \in [N]}} R^m_{i_1 \dots i_{m - 1} i_{m + 1} \dots i_N}.
\end{align*}
A tensor $A'$ that satisfies Equation~\eqref{eq:tensor_balancing} is called \emph{multistochastic}~\cite{Cui14}.
Note that this is exactly the same as the matrix balancing problem if $N = 2$.

It is straightforward to extend matrix balancing to tensor balancing as $e$-projection onto a submanifold.
Given a tensor $A \in \R^{n_1 \times n_2 \times \dots \times n_N}$ with its normalized probability distribution $P$ such that
\begin{align}
 \label{eq:normalize_tensor}
 p(x) = a_{i_1 i_2\dots i_N} \,\Big/\, \sum_{\mathclap{j_1 j_2 \dots j_N}} a_{j_1 j_2\dots j_N}
\end{align}
for all $x = (i_1, i_2, \dots, i_N)$.
The objective is to obtain $P_{\beta}$ such that $\sum_{i_m = 1}^n p_{\beta}( (i_1, \dots, i_N)) = 1/(n^{N - 1})$ for all $m \in [N]$ and $i_1, \dots, i_N \in [n]$.
In the same way as matrix balancing, we define $S$ as
\begin{align*}
 S = \Set{(i_1, i_2, \dots, i_N) \in [n]^N | a_{i_1 i_2\dots i_N} \not= 0}
\end{align*}
with removing zero entries and the partial order $\le$ as
\begin{align*}
 x = (i_1 \dots i_N) \le y = (j_1 \dots j_N) \Leftrightarrow \forall m \in [N], i_m \le j_m.
\end{align*}
In addition, we introduce $\vec{\iota}_{k, m}$ as
\begin{align*}
 \vec{\iota}_{k, m} = \min\Set{x = (i_1, i_2, \dots, i_N) \in S | i_m = k}.
\end{align*}
and require the condition in Equation~\eqref{eq:index_condition}.

\paragraph{Tensor balancing as \textit{e}-projection:}
Given a tensor $A \in \R^{n_1 \times n_2 \times \dots \times n_N}$ with its normalized probability distribution $P \in \vec{\Sc}$ given in Equation~\eqref{eq:normalize_tensor}.
The submanifold $\vec{\Sc}(\beta)$ of multistochastic tensors is given as
\begin{align*}
 \vec{\Sc}(\beta) = \Set{P \in \vec{\Sc} | \eta_P(x) = \beta(x) \text{ for all } x \in \dom(\beta)},
\end{align*}
where the domain of the function $\beta$ is given as
\begin{align*}
 \dom(\beta) = \Set{\vec{\iota}_{k, m} | k \in [n], m \in [N]}
\end{align*}
and each value is described using the zeta function as
\begin{align*}
 \beta(\vec{\iota}_{k, m}) = \sum_{l \in [n]}\zeta(\vec{\iota}_{k, m}, \vec{\iota}_{l, m}) \frac{1}{n^{N - 1}}.
\end{align*}
Tensor balancing is the $e$-projection of $P$ onto the submanifold $\vec{\Sc}(\beta)$, that is, the multistochastic tensor is the distribution $P_{\beta}$ such that
\begin{align*}
 \left\{
 \begin{array}{ll}
  \theta_{P_{\beta}}(x) = \theta_P(x) & \text{if } x \in S^+ \setminus \dom(\beta),\\
  \eta_{P_{\beta}}(x) = \beta(x) & \text{if } x \in \dom(\beta),
 \end{array}
 \right.
\end{align*}
which is unique and always exists in $\vec{\Sc}$, thanks to its dually flat structure.
Moreover, each balancing tensor $R^{m}$ is
\begin{align*}
 R^m_{i_1\dots i_{m - 1}i_{m + 1}\dots i_N} = \exp\left(\sum_{m' \not= m}\sum_{k = 1}^{i_{m'}} \theta_{P_\beta}(\vec{\iota}_{k, m'}) - \theta_{P}(\vec{\iota}_{k, m'})\right)
\end{align*}
for every $m \in [N]$ and $R^1 = R^1 n^{N - 1} / \sum_{j_1\dots j_N} a_{j_1\dots j_N}$ to recover a multistochastic tensor.\hfill$\blacksquare$
\vspace*{5pt}

Our result means that the $e$-projection algorithm based on Newton's method proposed in Section~\ref{subsec:proj} converges to the unique balanced tensor whenever $\vec{\Sc}(\beta) \not = \emptyset$ holds.

\section{Conclusion}\label{sec:conclusion}
In this paper, we have solved the open problem of tensor balancing and presented an efficient balancing algorithm using Newton's method.
Our algorithm quadratically converges, while the popular Sinkhorn-Knopp algorithm linearly converges.
We have examined the efficiency of our algorithm in numerical experiments on matrix balancing and showed that the proposed algorithm is several orders of magnitude faster than the existing approaches.

We have analyzed theories behind the algorithm, and proved that balancing is $e$-projection in a special type of a statistical manifold, in particular, a dually flat Riemannian manifold studied in information geometry.
Our key finding is that the gradient of the manifold, equivalent to Riemannian metric or the Fisher information matrix, can be analytically obtained using the M{\"o}bius inversion formula.

Our information geometric formulation can model several machine learning applications such as statistical analysis on a DAG structure.
Thus, we can perform efficient learning as projection using information of the gradient of manifolds by reformulating such models, which we will study in future work.

\section*{Acknowledgements}
The authors sincerely thank Marco Cuturi for his valuable comments.
This work was supported by JSPS KAKENHI Grant Numbers JP16K16115, JP16H02870 (MS), JP26120732 and JP16H06570 (HN).
The research of K.T.~was supported by JST CREST JPMJCR1502, RIKEN PostK, KAKENHI Nanostructure and KAKENHI JP15H05711.


\end{document}